\documentclass[11pt]{article}
\usepackage[margin=1in]{geometry}
\usepackage[T1]{fontenc}
\usepackage{times}
\usepackage{amsmath,amssymb,amsthm,mathtools,booktabs,array,graphicx,microtype}
\usepackage{natbib}
\usepackage{enumitem}
\setlist{nosep,leftmargin=*}
\usepackage{tikz}
\definecolor{lineA}{HTML}{0072B2}
\definecolor{lineB}{HTML}{E69F00}
\definecolor{lineC}{HTML}{009E73}
\usepackage[hidelinks]{hyperref}
\hypersetup{pdftitle={Randomization Beyond Deterministic Adaptivity in Parallel Sampling},pdfauthor={Da Li, Bo Jin, Xiaopeng Wei}}
\newtheorem{theorem}{Theorem}
\newtheorem{lemma}[theorem]{Lemma}
\newtheorem{corollary}[theorem]{Corollary}
\newtheorem{proposition}[theorem]{Proposition}
\newenvironment{keywords}{\begin{quote}\small\noindent\textbf{Keywords:}\ }{\end{quote}}
\newcommand{\KL}{\operatorname{D}}
\newcommand{\TV}{\operatorname{TV}}
\newcommand{\TC}{\operatorname{TC}}
\newcommand{\Aff}{\operatorname{Aff}}
\newcommand{\E}{\mathbb{E}}
\newcommand{\F}{\mathbb{F}}
\newcommand{\ind}{\mathbf{1}}
\newcommand{\ad}{\mathrm{ad}}
\newcommand{\mix}{\mathrm{mix}}
\newcommand{\Ber}{\operatorname{Ber}}
\newcommand{\logit}{\operatorname{logit}}
\newcommand{\sigmoid}{\operatorname{sigmoid}}

\title{Randomization Beyond Deterministic Adaptivity\\ in Parallel Sampling}
\author{Da Li \qquad Bo Jin\thanks{Co-corresponding authors: \texttt{jinbo@dlut.edu.cn} and \texttt{xpwei@dlut.edu.cn}.} \qquad Xiaopeng Wei\footnotemark[1]\\[4pt]
\small School of Computer Science and Technology, Dalian University of Technology\\
\small Dalian, Liaoning, China\\
\small\texttt{dli@mail.dlut.edu.cn}\quad\texttt{jinbo@dlut.edu.cn}\quad\texttt{xpwei@dlut.edu.cn}}
\date{}
\begin{document}
\maketitle
\begin{abstract}
Parallel samplers, including masked diffusion models, generate several coordinates per round from conditional marginals, and their accuracy depends on which coordinates are revealed together. A schedule can be improved in two natural ways: adaptively, by choosing later batches from the values already generated, or by randomization, by drawing the schedule at random. We show that randomization is strictly more powerful. For every $r\ge1$ we construct a full-support distribution on $9r$ binary coordinates such that every deterministic three-round sampler, even with value-dependent batches and arbitrary product kernels, has forward Kullback--Leibler divergence at least $\gamma r-\log4$ with $\gamma\approx0.0966$, whereas drawing one of four fixed schedules with two fair bits drives the divergence below $18r\beta^r$ with $\beta\approx0.966$. Both samplers query conditionals of the same target. The separation holds for every fixed number of rounds $T\ge2$, a strict advantage survives perturbations of the target and errors in the conditional predictor, and exact rational arithmetic certifies the gap at nine coordinates, where trained neural predictors preserve it. The mechanism is a hidden mixture: each schedule preserves the dependence pattern of one component, while no deterministic first batch can serve them all.
\end{abstract}


\section{Introduction}
Masked diffusion models and related parallel samplers generate a high-dimensional sample in a small number of rounds. In each round a batch of coordinates is drawn independently from conditional marginals given the coordinates generated so far, and committed values are never revised \citep{chen,lavenant,secchi}. Fewer rounds mean faster generation, but coordinates that remain dependent given the past are drawn as if they were independent. The resulting error is governed by the \emph{schedule}, the choice of which coordinates are revealed together.

There are two natural ways to choose a good schedule. An \emph{adaptive} rule selects each batch as a function of the values generated so far, as in confidence-based decoding \citep{cai}. For deterministic adaptive rules, \citet{wen} give an exact accounting: the forward divergence equals the expected conditional total correlation discarded along the reveal tree. A \emph{randomized} sampler instead draws a schedule at random and executes it once, so its output is a mixture of the laws of fixed schedules. By convexity this mixture is never worse than the average schedule \citep[Theorem~3.3]{chen}. Whether it can be much better than every deterministic rule, adaptive or not, is a different question, raised by \citet[Remark~2]{wen}:
\begin{center}
\emph{Can randomizing the schedule achieve an accuracy that no deterministic schedule,\\ however adaptive, can match?}
\end{center}
We answer this question affirmatively and quantitatively: at three rounds, the deterministic error grows linearly with the dimension, while the randomized error decays exponentially.

\begin{theorem}[Main result, informal]\label{thm:informal}
There are constants $\gamma\approx0.0966$ and $\beta\approx0.966$ such that for every $r\ge1$ there is a full-support distribution $p_r$ on $\{0,1\}^{9r}$ with the following two properties.
\begin{enumerate}
\item The output law $q$ of every deterministic sampler that completes in at most three rounds satisfies $\KL(p_r\|q)\ge\gamma r-\log4$, even when its batches are chosen adaptively from the generated values and its kernels are arbitrary product kernels.
\item The output law $q_\mix$ obtained by drawing one of four fixed three-round schedules with two fair bits, using the exact conditional marginals of $p_r$, satisfies $\KL(p_r\|q_{\mix})\le18r\beta^r$.
\end{enumerate}
\end{theorem}

Theorem~\ref{thm:three} states the result formally, and Figure~\ref{fig:construction} shows the construction. Both samplers query conditional probabilities of the same known target, and the randomized sampler spends only two random bits per sample beyond the randomness of the coordinate draws; ``deterministic'' refers to the choice of batches.

\paragraph{What the separation shows.}
Adaptivity and randomization are different resources. Adaptivity uses observations to choose future batches; randomization changes the output law by averaging over schedules. Theorem~\ref{thm:informal} shows that the first cannot substitute for the second, even with exact conditional probabilities and three rounds. Two consequences are worth stating. First, lower bounds proved for deterministic adaptive rules, for instance through conditional total correlation, do not carry over to randomized samplers, so the random orders used by masked diffusion models call for an analysis of the mixture itself. Second, the convexity bound of \citet{chen} can be loose by an unbounded factor: each of the four schedules in Theorem~\ref{thm:informal} is itself a deterministic rule, with divergence at least $\gamma r-\log4$, yet their uniform mixture has exponentially small divergence.

\paragraph{The mechanism: a hidden mixture.}
The target is a uniform mixture of four distributions with identical one-coordinate marginals and different dependence patterns; the component label is not part of the output, so every conditional the sampler queries averages over it. Each of the four schedules is tailored to one component. Under that component its batches are conditionally independent, so the schedule reproduces the component; under the other components its first batch is correlated, which lets the generated values reveal the label. Mixing the four schedules then reproduces the whole target. A deterministic rule faces a dilemma at its first batch, which it must commit to before seeing any value. The four dependence patterns are incompatible, and a short finite-geometry argument over $\F_3^2$ shows that every first batch discards dependence in a constant fraction of coordinate groups, whatever the later adaptive choices.

\paragraph{Further results.}
\begin{itemize}
\item \emph{Sharpness.} The deterministic optimum equals $\gamma r-\log4$ up to an exponentially small term, and the separation also holds in total variation (Theorem~\ref{thm:three}).
\item \emph{Every fixed depth.} For each fixed number of rounds $T\ge2$ there are targets with linear deterministic error and exponentially small randomized error (Theorems~\ref{thm:two} and~\ref{thm:allT}). One round admits no advantage.
\item \emph{All four schedules are needed.} Uniform weights achieve exponential accuracy, unequal positive weights $\pi$ leave exactly the floor $\KL(u_4\|\pi)$, with $u_4$ the uniform weights, and omitting any schedule restores a linear error even after reoptimizing every kernel (Proposition~\ref{prop:weights}).
\item \emph{Robustness.} A strict advantage survives perturbations of the component laws and errors in the conditional predictor, with explicit thresholds (Corollary~\ref{cor:oracle} and Proposition~\ref{prop:perturb}).
\item \emph{Exact certificates.} At nine coordinates, exact rational interval arithmetic certifies a deterministic optimum of $0.0781127$ nats against $0.0702293$ for the four-schedule mixture, and neural predictors fitted to the target's conditionals preserve the gap (Section~\ref{sec:verification}).
\end{itemize}

\paragraph{Related work.}
Error analyses of parallel sampling quantify the cost of independent coordinate updates through information profiles and conditional total correlation \citep{chen,lavenant,secchi,wen}. \citet{licai} bound the error of diffusion language models by the mutual information between tokens divided by the number of iterations; \citet{zhaocai} randomize the number of tokens revealed per step and obtain distribution-agnostic bounds in terms of total and dual total correlation; \citet{wainwright} certifies adaptive schedules through unmasking growth complexity; and \citet{zhang2026} separate order sensitivity from parallelization bias. These works bound or optimize the error of a given class of schedules, whereas we compare two classes. Our result also complements the randomized sizes of \citet{zhaocai}: randomizing the schedule can be strictly necessary even when the target is known.

\citet{chen} characterize optimal schedules through the information curve and distinguish divergence averaged over schedules from divergence of the final mixture; \citet{wen} derive the exact total-correlation identity for deterministic adaptive rules. Our construction gives a quantitative separation between these two classes of samplers. The first arXiv version of this work \citep{lijinwei} used a construction in which the component label was among the output coordinates; the present target outputs only data, so every conditional query averages over the unobserved label. In concurrent work, \citet{wu} studies the number of rounds needed to reach a fixed accuracy, including separations on balanced slices, full-support extensions and lower bounds for causal randomized policies. We study accuracy at a fixed depth: a constant-size library attains exponentially small error on a full-support target whose component label is never observed. \citet{wucai} show that transformer-based masked diffusion models learn mixtures of hidden Markov trees under parallel unmasking; our targets are also latent-variable mixtures, and for them the choice between deterministic and randomized schedules matters even with exact conditionals. Confidence-based decoding \citep{cai} analyzes particular adaptive rules, and autospeculative sampling \citep{anari} uses rejection, a different sampling interface.

\paragraph{Organization.}
Section~\ref{sec:tools} formalizes the two classes of samplers. Section~\ref{sec:three} proves the three-round separation, Section~\ref{sec:depth} extends it to every fixed depth, and Section~\ref{sec:oracle} establishes robustness. Section~\ref{sec:verification} reports exact certificates and learned predictors. Remaining proofs are in the appendix.

\section{Sampling model}\label{sec:tools}
All spaces are finite and logarithms are natural. We use the Kullback--Leibler (KL) divergence, the total-variation (TV) distance and the affinity (Bhattacharyya coefficient)
\[
\KL(p\|q)=\sum_xp(x)\log\frac{p(x)}{q(x)},\quad
\TV(p,q)=\tfrac12\sum_x|p(x)-q(x)|,\quad
\Aff(p,q)=\sum_x\sqrt{p(x)q(x)}.
\]

\paragraph{Deterministic adaptive samplers.}
A deterministic policy chooses a nonempty set $B(h)$ of unrevealed coordinates from the complete committed history $h$ and commits independent draws from scalar kernels $K_{j,h}$; committed values are never revised. A budget of $T$ rounds requires completion on every trajectory. With exact kernels, $K_{j,h}=p(X_j\mid h)$. The output law is $q_{A,K}$; the kernels need not arise from a common joint distribution.

\paragraph{Randomized fixed samplers.}
A fixed schedule is an ordered partition chosen before generation. A randomized fixed sampler draws one schedule $s$ with probability $\pi_s$, independently of the coordinate draws, executes it once and returns the completed vector, so its law is $\sum_s\pi_s q_s$. Both classes know the target.

\paragraph{The reveal-tree identity.}
At a history $h$, the conditional total correlation of a batch is
\[
\TC_P(X_B\mid h)=\KL\left(P(X_B\mid h)\middle\|\prod_{j\in B}P(X_j\mid h)\right).
\]
For any reference distribution $P$, every deterministic tree satisfies
\begin{equation}\label{eq:tree}
\KL(P\|q_{A,K})=
\E_P\sum_t\left\{\TC_P(X_{B_t}\mid h_t)
 +\sum_{j\in B_t}\KL\bigl(P(X_j\mid h_t)\|K_{j,h_t}\bigr)\right\}.
\end{equation}
A completed vector determines a unique path; factoring both laws along that path and inserting the product of reference scalar conditionals proves the identity. In particular, exact target marginals minimize forward KL for each fixed deterministic rule. This is the kernel-error extension of the identity of \citet{wen}.

\paragraph{Mixture targets.}
Let $p(x)=\sum_{z=1}^K\nu_zP_z(x)$ with $\nu_z>0$. The latent index $Z$ is only a mathematical representation of $p$; it is never observed by either sampler.
\begin{lemma}[Componentwise divergence]\label{lem:latent}
For every output law $q$,
\begin{equation}\label{eq:discount}
\KL(p\|q)=\sum_z\nu_z\KL(P_z\|q)-I(Z;X).
\end{equation}
Consequently, any lower bound $L$ on the average component divergence yields
$\KL(p\|q)\ge L-H(\nu)$.
\end{lemma}
\begin{proof}
Expand the average component log likelihood ratio and insert $p(x)$. The additional term is $\sum_z\nu_z\KL(P_z\|p)=I(Z;X)$, and $I(Z;X)\le H(\nu)$.
\end{proof}
Combined with \eqref{eq:tree}, the lemma reduces lower bounds for the mixture to the average componentwise cost of the same sampler, which never observes the label, at a price of at most the label entropy.

\section{The three-round separation}\label{sec:three}
\begin{figure}[t]
\centering
\begin{tikzpicture}[font=\small,
  dot/.style={circle,inner sep=0pt,minimum size=2.7mm,draw=black!60,line width=0.35pt},
  big/.style={circle,inner sep=0pt,minimum size=4.6mm,line width=0.5pt,font=\scriptsize}]
\foreach \k/\name/\dx in {0/{(0,1)}/0, 1/{(1,0)}/2.0, 2/{(1,1)}/4.0, 3/{(1,2)}/6.0}{
  \foreach \x in {0,1,2}{\foreach \y in {0,1,2}{
     \pgfmathtruncatemacro{\ci}{\k==0 ? \x : (\k==1 ? \y : (\k==2 ? mod(\y-\x+3,3) : mod(\x+\y,3)))}
     \ifcase\ci \def\col{lineA}\or \def\col{lineB}\or \def\col{lineC}\fi
     \node[dot,fill=\col] at ({\dx+0.5*\x},{0.5*\y}) {};
  }}
  \node[font=\footnotesize] at ({\dx+0.5},-0.5) {$z=\name$};
}
\node[font=\footnotesize] at (3.75,-1.15) {(a) the four components; colors mark the lines of direction $z$};
\begin{scope}[xshift=9.4cm]
  \fill[black!10,rounded corners=4pt] (-0.42,0.42) rectangle (1.92,1.08);
  \foreach \x/\y/\t in {0/0/1,0/1/2,0/2/3,1/0/3,1/1/1,1/2/2,2/0/3,2/1/1,2/2/2}{
     \ifcase\x \def\col{lineA}\or \def\col{lineB}\or \def\col{lineC}\fi
     \node[big,draw=\col,fill=\col!20] at ({0.75*\x},{0.75*\y}) {\t};
  }
  \node[font=\scriptsize,anchor=west,align=left,text=black!75] at (2.05,0.75) {a line of\\ $w=(1,0)$};
  \node[font=\footnotesize] at (1.05,-1.15) {(b) rounds of the schedule for $s=(0,1)$};
\end{scope}
\end{tikzpicture}
\caption{The construction on $V=\F_3^2$. (a) Each direction $z\in\mathcal D$ partitions $V$ into three parallel lines, shown in three colors; under component $P_z$, every line carries an independent copy of the law $R$ in~\eqref{eq:R}. (b) The schedule assigned to $s=(0,1)$ reveals the slices $C_s$, $D_s$ and $E_s$ of Section~\ref{sec:four} in rounds $1$, $2$ and $3$. Every $s$-line (one color) meets each round once, so the schedule factorizes $P_s$. The shaded line of the wrong direction $w=(1,0)$ has two first-round points, a correlated pair under $P_w$ that reveals the hidden component.}
\label{fig:construction}
\end{figure}
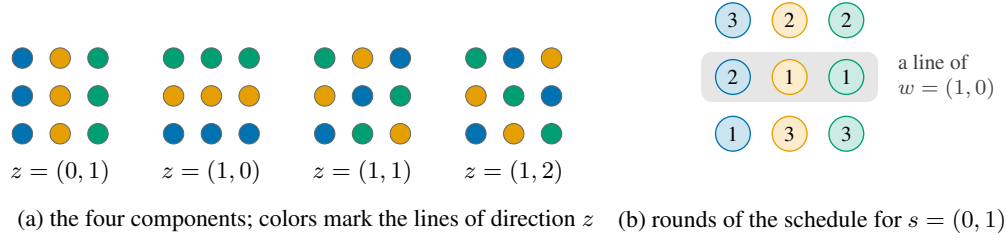

\paragraph{The construction.}
Let $V=\F_3^2$ and $\mathcal D=\{(0,1),(1,0),(1,1),(1,2)\}$. Each direction partitions $V$ into three affine lines. On each line, independently, put the full-support law
\begin{equation}\label{eq:R}
R(000)=R(111)=5/16,\qquad R(x)=1/16\quad\text{otherwise}.
\end{equation}
Denote the resulting nine-coordinate distribution for direction $z$ by $P_z$. The target distribution is
\begin{equation}\label{eq:dataonly}
p_r(x^{1:r})=\frac14\sum_{z\in\mathcal D}\prod_{\ell=1}^rP_z(x^\ell),\qquad d=9r.
\end{equation}
Only the $9r$ data coordinates are output, and every queried conditional is a conditional of \eqref{eq:dataonly}. Define
\begin{align}
a&=\tfrac34\log3-\log2, &
b&=\log2-\tfrac32\log3+\tfrac58\log5,\\
\gamma&=\tfrac14(a+5b)=0.09661231252165\ldots,&
\beta&=\tfrac12\left(\sqrt{3/2}+\sqrt{1/2}\right),\\
\sigma&=\left\|\sqrt{\begin{pmatrix}5/8&1/8\\1/8&1/8\end{pmatrix}}\right\|_{\mathrm{op}}<1,
\end{align}
where the square root of the matrix is entrywise. Let $E_\ad(p,T)$ be the infimum over deterministic exact-marginal policies with round budget $T$; by \eqref{eq:tree} it equals the optimum over arbitrary product kernels.

\begin{theorem}[Three-round separation]\label{thm:three}
For every $r\ge1$:
\begin{align}
\inf_{A,K:\,T\le3}\KL(p_r\|q_{A,K})&\ge r\gamma-\log4,\label{eq:three-lower}\\
0\le E_\ad(p_r,3)-(r\gamma-\log4)
&\le 3\bigl[7r(1+\log2)+1\bigr]\beta^{2\lfloor r/4\rfloor},\label{eq:three-sharp}\\
\KL(p_r\|q_{\mix,r})&\le18r\beta^r,\label{eq:three-upper}\\
\inf_{A,K:\,T\le3}\TV(p_r,q_{A,K})
&\ge\tfrac14-\tfrac12\sigma^{3r/4}.\label{eq:three-tv}
\end{align}
The mixture in \eqref{eq:three-upper} uses four fixed three-round schedules and two fair scheduling bits. The lower bounds allow arbitrary product kernels. At $r=1$, exact rational interval computations certify
\begin{equation}\label{eq:nine}
\KL(p_1\|q_{\mix,1})<0.0702292963221
<0.0781126887145<E_\ad(p_1,3).
\end{equation}
\end{theorem}

\paragraph{Proof outline.}
The proof has three steps. First, every deterministic first batch discards at least $\gamma$ nats of dependence per block, averaged over the four hidden components and whatever the later adaptive choices (Lemma~\ref{lem:firstbatch}); with Lemma~\ref{lem:latent} this gives \eqref{eq:three-lower}. Second, a fixed schedule that factorizes one component pays only for its uncertainty about the hidden label, and the affinity of its first batch controls that uncertainty (Lemma~\ref{thm:posterior}). Third, for each direction $s$, three parabolic slices of $V$ form a schedule that factorizes $P_s$ and exposes a correlated pair under every other component (Lemma~\ref{lem:slices} and Figure~\ref{fig:construction}b); the last two steps give \eqref{eq:three-upper}. Appendix~\ref{app:three} proves \eqref{eq:three-sharp} and \eqref{eq:three-tv}, Appendix~\ref{app:cert} proves \eqref{eq:nine}, and Figure~\ref{fig:bounds} plots both analytic bounds.

\subsection{Every deterministic first batch is costly}\label{sec:lower}
For $C\subseteq V$, let $N_j(C)$ count the affine lines, over all four directions, that meet $C$ in exactly $j$ points, and put
\[
F(C)=\tfrac14\{bN_0(C)+aN_2(C)+(2a+b)N_3(C)\}.
\]
\begin{lemma}[Cost of a first batch]\label{lem:firstbatch}
For every $C\subseteq V$, $F(C)\ge\gamma$, with equality exactly when $|C|=2$. If a deterministic sampler with round budget three reveals $C_\ell$ in block $\ell$ in its first round, then for all product kernels
\[
\frac14\sum_{z\in\mathcal D}\KL(P_z^{\otimes r}\|q_{A,K})\ge\sum_{\ell=1}^rF(C_\ell)\ge r\gamma.
\]
\end{lemma}
\begin{proof}
For the line law $R$, the unconditional pair mutual information is $a$, either conditional pair mutual information given the third bit is $b$, and $\TC(R)=2a+b$. If a line contains $j$ first-round points, its unavoidable dependence cost in the remaining two rounds is
\[
\begin{array}{c|cccc}j&0&1&2&3\\\hline
\text{cost}&b&0&a&2a+b.
\end{array}
\]
For $j=0$, three untouched coordinates must be placed into at most two further rounds; for $j\ge2$, the first round already discards dependence. Conditional on $z$, distinct lines of direction $z$ are independent, so observations in other lines or blocks do not change the conditional law of a line once its own observed bits are fixed. Since $F(C)$ is the average over $z$ of the costs of the lines of direction $z$, applying \eqref{eq:tree} under each $P_z^{\otimes r}$ gives the displayed bound.

At cardinalities $0,1,2$, $F$ equals $3b$, $2b$ and $(a+5b)/4=\gamma$. At cardinality three it equals $(2a+3b)/4$ for collinear points and $(3a+3b)/4$ otherwise. These values other than $\gamma$ exceed $\gamma$ since $2b<a<3b$. If $t=|C|\ge4$, every pair of points of $C$ lies on a unique line, so $N_2+3N_3=\binom t2$ and $F(C)\ge\frac a4(N_2+2N_3)\ge\frac a6\binom t2\ge a>\gamma$.
\end{proof}
Lemmas~\ref{lem:latent} and~\ref{lem:firstbatch} together give \eqref{eq:three-lower}.

\subsection{Fixed schedules identify the hidden component}\label{sec:posterior}
The upper bound rests on a general principle: a fixed schedule that factorizes one component exactly pays only for its uncertainty about the hidden label, and this uncertainty is controlled by the affinity of the first batch. Consider any finite mixture $p=\sum_z\nu_zP_z$ on $d$ coordinates, with $\nu_z>0$. For each $z$, assign a fixed schedule $s_z=(B_{z,1},\ldots,B_{z,T_z})$. Set $A_z=B_{z,1}$, $N_z=d-|A_z|$, and
\begin{align}
\Delta_z&=\sum_t\E_{P_z}\TC_{P_z}(X_{B_{z,t}}\mid X_{B_{z,<t}}),\\
 a_z&=\sum_{j\in A_z}\KL(P_z(X_j)\|p(X_j)),\\
 b_z&=\sum_{w\ne z}\sqrt{\frac{\nu_w}{\nu_z}}\,
 \Aff\bigl((P_z)_{A_z},(P_w)_{A_z}\bigr).
\end{align}
Let $q_z$ be the output of $s_z$ using the exact scalar marginals of $p$, so that every query integrates over the latent component.
\begin{lemma}[Error bound from component identification]\label{thm:posterior}
For the preceding schedules,
\begin{equation}\label{eq:posterior}
\KL\left(p\middle\|\sum_z\nu_zq_z\right)
\le \sum_z\nu_z\KL(P_z\|q_z)
\le \sum_z\nu_z(\Delta_z+a_z+N_zb_z).
\end{equation}
The component laws may have arbitrary dependence.
\end{lemma}
\begin{proof}
The first inequality is KL data processing from joint laws with the same latent prior and conditional laws $P_z,q_z$. For a history $h$ along $s_z$, let $\theta_z(h)=\Pr_p(Z=z\mid h)$. Every scalar mixture conditional obeys
\[
p(X_j=u\mid h)\ge\theta_z(h)P_z(X_j=u\mid h).
\]
Hence the scalar kernel divergence in \eqref{eq:tree} is at most $-\log\theta_z(h)$, and first-round scalar errors are exactly $a_z$. For nested observed prefixes $A\subseteq H$, Bayes' rule gives
\begin{equation}\label{eq:posterior-monotone}
\E_{P_z}[-\log\theta_z(H)]
=\E_{P_z}[-\log\theta_z(A)]
-\E_{(P_z)_A}\KL(P_z(X_{H\setminus A}\mid X_A)\|p(X_{H\setminus A}\mid X_A)),
\end{equation}
so the expected posterior log loss never increases. At the first prefix,
\[
-\log\theta_z(A_z)=\log\left(1+\sum_{w\ne z}\frac{\nu_w(P_w)_{A_z}(X_{A_z})}{\nu_z(P_z)_{A_z}(X_{A_z})}\right).
\]
The inequalities $\log(1+\sum_i u_i)\le\sum_i\log(1+u_i)\le\sum_i\sqrt{u_i}$ bound its expectation by $b_z$. Summing over the $N_z$ later coordinates proves \eqref{eq:posterior}.
\end{proof}

\subsection{Four schedules, one per component}\label{sec:four}
For each $s\in\mathcal D$, choose $v_s=(1,0)$ when $s=(0,1)$ and $v_s=(0,1)$ otherwise, and set
\[
C_s=\{tv_s+t^2s:t\in\F_3\},\qquad D_s=C_s+s,\qquad E_s=C_s+2s.
\]
\begin{lemma}[Parabolic slices]\label{lem:slices}
For each $s\in\mathcal D$, the slices $C_s,D_s,E_s$ partition $V$, and every line of direction $s$ meets each slice exactly once. For every $w\in\mathcal D\setminus\{s\}$, the three lines of direction $w$ meet $C_s$ in $2$, $1$ and $0$ points.
\end{lemma}
\begin{proof}
The lines of direction $u$ are the level sets of a nonzero linear functional $\lambda_u$ with $\lambda_u(u)=0$. As $v_s$ is not parallel to $s$, $t\mapsto\lambda_s(tv_s+t^2s)=t\lambda_s(v_s)$ is a bijection of $\F_3$, so every $s$-line meets $C_s$ once; translation by $s$ preserves $s$-lines, so the three translates partition each $s$-line. For $w\ne s$, $\lambda_w(s)\ne0$, so $\lambda_w(tv_s+t^2s)$ is a quadratic in $t$ with nonzero leading coefficient; on $\F_3$ such a quadratic takes one value twice and another once.
\end{proof}
The schedule assigned to $s$ reveals $C_s^{1:r}$, then $D_s^{1:r}$, then $E_s^{1:r}$. Under $P_s^{\otimes r}$, the coordinates in each slice are conditionally independent given the previous slices, and the first slice consists of $3r$ independent fair bits; hence $\Delta_s=a_s=0$. Under every wrong component, the first slice of each block contains a pair with probability matrix
\[
\begin{pmatrix}3/8&1/8\\1/8&3/8\end{pmatrix}
\]
and one independent fair bit, whose affinity with the uniform law of the first slice is $\beta$. Lemma~\ref{thm:posterior}, with $\Delta_s=a_s=0$, $N_s=6r$ and $b_s\le3\beta^r$, proves \eqref{eq:three-upper}. The bound already accounts for the residual uncertainty about the component in both later rounds.

\subsection{All four schedules are necessary}\label{sec:weights}
The four schedules approximate different components, so their weights matter. Fix the four schedules above, let $u_4$ be the uniform two-fair-bit rule, let $\pi$ be arbitrary selection probabilities, and write $q_{\pi,K,r}$ for the mixture with arbitrary scalar kernels; the subscript $\mathrm{exact}$ denotes the exact marginals of $p_r$.

\begin{proposition}[Schedule weights]\label{prop:weights}
For every $r\ge1$ and every strictly positive $\pi$, define
\[
\varepsilon_r(\pi)=\frac12\sum_z\log\left(1+\beta^r
                   \sum_{s\ne z}\sqrt{\pi_s/\pi_z}\right).
\]
Then
\begin{align}
\KL(u_4\|\pi)-\varepsilon_r(\pi)
&\le \inf_K\KL(p_r\|q_{\pi,K,r})\label{eq:weights-lower}\\
&\le\KL(p_r\|q_{\pi,\mathrm{exact},r})
 \le\KL(u_4\|\pi)+18r\beta^r.\label{eq:weights-upper}
\end{align}
If $\pi$ is supported on $m\le3$ of these schedules, then, for every choice of kernels,
\begin{equation}\label{eq:missing-weight}
\KL(p_r\|q_{\pi,K,r})\ge
\frac{4-m}{2}r\log(1/\beta)-\frac{4-m}{4}\log m-\log4.
\end{equation}
\end{proposition}

Uniform weights therefore achieve exponentially small error, fixed unequal positive weights converge to the floor $\KL(u_4\|\pi)$ whether or not the kernels are reoptimized, and dropping any schedule restores a linear error. The key fact is that, for $s\ne z$, the first batch of schedule $s$ has affinity at most $\beta^r$ with component $z$ under every product kernel, and later draws cannot repair this mismatch. Appendix~\ref{app:weights} gives the proof.

\section{Every fixed number of rounds}\label{sec:depth}
\subsection{Two rounds}\label{sec:two}
A two-round policy makes no value-dependent choice after its first batch, since completion forces the second batch to be the complement. Let $\mathcal M_0=\{01,23\}$, $\mathcal M_1=\{02,13\}$ and $\mathcal M_2=\{03,12\}$ be the three perfect matchings on four vertices. For $0<\eta<1/2$, put on each edge the law
\[
R_\eta(00)=R_\eta(11)=(1-\eta)/2,\qquad
R_\eta(01)=R_\eta(10)=\eta/2.
\]
Conditional on a uniform latent $z\in\{0,1,2\}$, take $2r$ independent four-bit blocks with matching $\mathcal M_z$, and output only the $8r$ data bits. Denote the mixture by $p_{r,\eta}$ and put
\[
I_\eta=\log2-h_2(\eta),\qquad
\beta_\eta=\frac{\sqrt{2(1-\eta)}+\sqrt{2\eta}}2<1,
\]
where $h_2(t)=-t\log t-(1-t)\log(1-t)$ is the binary entropy in nats.

\begin{theorem}[Two-round separation]\label{thm:two}
Every deterministic two-round product sampler obeys
\begin{align}
\KL(p_{r,\eta}\|q)&\ge\tfrac43rI_\eta-\log3,\\
\TV(p_{r,\eta},q)&\ge\tfrac13-\tfrac1{\sqrt3}\beta_\eta^{4r/3}.
\end{align}
Three fixed schedules have uniform-mixture divergence at most $8r\beta_\eta^r$. At $r=1,\eta=1/4$, a uniform mixture of just two of these schedules, selected by one fair bit, satisfies
\begin{equation}\label{eq:eight}
\KL(p_{1,1/4}\|q_{\mathrm{bit}})<0.0767188675106
<0.0918463303897<\inf_{A,K:T\le2}\KL(p_{1,1/4}\|q_{A,K}).
\end{equation}
No randomization advantage is possible in one round when every component uses exact target marginals.
\end{theorem}
The proof and explicit schedules are in Appendix~\ref{app:two}.

\subsection{Every fixed depth}
Fix $T\ge2$. Let $\mathcal Z$ be the unordered partitions of $[2T]$ into two sets of size $T$, and let $K=\frac12\binom{2T}{T}$. For $0<\epsilon<1/2$, define an exchangeable $T$-bit law $R_T$ by drawing a hidden fair bit and passing it through $T$ independent binary symmetric channels with crossover $\epsilon$. Under component $z$, put independent $R_T$ laws on its two parts. Repeat this $2T$-coordinate block $r$ times with the same uniform latent partition, and call the marginal data law $p_{T,r}$. Let $a_T$ be the minimum mutual information of two unobserved $R_T$ coordinates, conditional on any assignment to any subset of the other coordinates. Then $a_T>0$: the posterior of the hidden bit stays strictly between zero and one, the conditional covariance of two remaining noisy copies is strictly positive, and there are finitely many conditioning assignments. Write $\eta=2\epsilon(1-\epsilon)$ and $\beta=\beta_\eta$.

\begin{theorem}[Every fixed depth]\label{thm:allT}
For every fixed $T\ge2$ and every $r\ge1$,
\begin{align}
\inf_{A,K':\,\text{rounds}\le T}\KL(p_{T,r}\|q_{A,K'})
&\ge\frac{2(T-1)}{2T-1}a_Tr-\log K,\label{eq:allTlower}\\
\KL(p_{T,r}\|q_\mix)
&\le2(T-1)r(K-1)\beta^{\lfloor r/(K-1)\rfloor}.\label{eq:allTupper}
\end{align}
The mixture uses $K$ fixed $T$-round schedules. Thus, for each fixed budget, deterministic error is at least linear in dimension whereas mixture error vanishes exponentially.
\end{theorem}
The proof is in Appendix~\ref{app:depth}.

\section{Robustness}\label{sec:oracle}
The randomized construction does not rely on exact conditional probabilities. For binary outputs, its error is controlled by the expected squared difference between predicted and exact logits.

\begin{lemma}[Bernoulli logit perturbation]\label{lem:logit}
For $u\in[0,1]$, $v=\sigmoid(b)$ and $\widehat v=\sigmoid(b+e)$,
\begin{equation}\label{eq:logit}
\KL(\Ber(u)\|\Ber(\widehat v))
\le2\KL(\Ber(u)\|\Ber(v))+\tfrac14 e^2.
\end{equation}
\end{lemma}
The proof is in Appendix~\ref{app:logit}. For the schedule assigned to $z$, define the target-history error
\begin{equation}\label{eq:E2}
\mathcal E_2=\sum_z\nu_z\E_{X\sim P_z}
\sum_{t,j\in B_{z,t}}
\left[\logit K_{z,j,h_t}(1)-\logit p(X_j=1\mid h_t)\right]^2,
\end{equation}
where histories are obtained by revealing target samples along the assigned schedule and all probabilities lie strictly between zero and one.

\begin{corollary}[Oracle-error budget]\label{cor:oracle}
Suppose that Lemma~\ref{thm:posterior} gives $\sum_z\nu_z\KL(P_z\|q_z)\le B$, an average over components of the divergence between each component and the output of its own schedule. Then the same mixture with approximate kernels satisfies
\begin{equation}\label{eq:robust}
\KL(p\|q_{\mix,K})\le2B+\mathcal E_2/4.
\end{equation}
For the three-round family, if $\mathcal E_2\le9r\delta^2$, then
\begin{equation}\label{eq:robust3}
\KL(p_r\|q_{\mix,K})\le36r\beta^r+\tfrac94r\delta^2.
\end{equation}
In particular, any fixed $\delta^2<4\gamma/9$ preserves an asymptotic strict advantage over every deterministic three-round product sampler.
\end{corollary}
\begin{proof}
Apply Lemma~\ref{lem:logit} to each scalar kernel term in the component version of \eqref{eq:tree}. If the conditional total-correlation sum is $\Delta_z$, the result is at most $2\KL(P_z\|q_z)-\Delta_z$ plus one quarter of the squared-error sum. Drop $-\Delta_z$, average over $z$ and use mixture data processing. The deterministic lower bound allows arbitrary product kernels and is unchanged.
\end{proof}

\paragraph{Perturbed components.}
Equal scalar marginals and exact within-component factorization drive the exponential bound, and a strict advantage persists when they hold only approximately. For $0\le\epsilon<1/4096$, put
\begin{align}
\omega(\epsilon)&=18\epsilon\log2+10h_2(\epsilon), &
\alpha(\epsilon)&=\frac{12\epsilon^2}{1/4-\epsilon^2},&
\rho(\epsilon)&=\beta+2\sqrt{2\epsilon}.
\end{align}
\begin{proposition}[Stability under component perturbations]\label{prop:perturb}
Let $\widetilde P_z$ be nine-coordinate laws with
$\max_z\TV(\widetilde P_z,P_z)\le\epsilon$, assume $\rho(\epsilon)<1$, and set
$\widetilde p_r=\frac14\sum_z\widetilde P_z^{\otimes r}$, a full-support law.
Every deterministic three-round product sampler satisfies
\begin{equation}\label{eq:perturb-lower}
\KL(\widetilde p_r\|q_{A,K})\ge r[\gamma-\omega(\epsilon)]-\log4.
\end{equation}
The uniform mixture of the four schedules, using exact marginals of $\widetilde p_r$, satisfies
\begin{equation}\label{eq:perturb-upper}
\KL(\widetilde p_r\|\widetilde q_{\mix,r})
\le r[\omega(\epsilon)+\alpha(\epsilon)]+18r\rho(\epsilon)^r.
\end{equation}
For finite-logit approximate kernels, with $\widetilde{\mathcal E}_2$ defined as in \eqref{eq:E2} for target $\widetilde p_r$ and components $\widetilde P_z^{\otimes r}$,
\begin{equation}\label{eq:perturb-learned}
\KL(\widetilde p_r\|\widetilde q_{\mix,K,r})
\le r[\omega(\epsilon)+2\alpha(\epsilon)]
       +36r\rho(\epsilon)^r+\widetilde{\mathcal E}_2/4.
\end{equation}
\end{proposition}
A strict asymptotic advantage therefore holds whenever $2\omega(\epsilon)+\alpha(\epsilon)<\gamma$, and with approximate kernels satisfying $\widetilde{\mathcal E}_2\le9r\delta^2$ whenever $2\omega(\epsilon)+2\alpha(\epsilon)+9\delta^2/4<\gamma$. For example, at $\epsilon\le10^{-5}$ and $\delta\le0.1$ the bounds certify, for every $r\ge256$, a deterministic error above $0.0898r$ and a randomized error below $0.0794r$. The proof and a rational certificate for this example are in Appendix~\ref{app:perturb}.

\section{Exact certificates and learned predictors}\label{sec:verification}
Our programs compute every deterministic optimum by exhaustive dynamic programming on partial assignments, over all first batches and all value-dependent later batches, and enclose it with rational probabilities and certified logarithm intervals (Appendix~\ref{app:cert}).

We also fit a shared neural predictor to the exact conditional logits of the nine-coordinate target, without access to the component label, freeze it, and reconstruct its output law over all 512 binary vectors (three training seeds; Appendix~\ref{app:training}). Table~\ref{tab:learned} compares the geometric library of Section~\ref{sec:four} with an independently drawn four-schedule library and with uniform sampling over all 1680 ordered partitions with batch sizes $(3,3,3)$. Every randomized method executes one schedule and makes three predictor calls per output.

\begin{table}[htbp]
\centering\small
\setlength{\tabcolsep}{3.5pt}
\begin{tabular}{@{}>{\raggedright\arraybackslash}p{.33\linewidth}rrl@{}}
\toprule
Sampler & Exact-oracle KL & Learned KL & Scheduling randomness\\
\midrule
\multicolumn{4}{@{}l}{\emph{Nine coordinates, three rounds}}\\
All deterministic adaptive policies & $0.07811269$ & --- & none\\
Geometric four-schedule library & $0.07022930$ & $0.07180678\pm0.00026039$ & 2 fair bits\\
Random four-schedule library & $0.06685022$ & $0.06837191\pm0.00009104$ & 2 fair bits\\
Uniform $(3,3,3)$ partitions & $0.04952470$ & $0.05026709\pm0.00007323$ & $\log_2 1680$ bits\\
\midrule
\multicolumn{4}{@{}l}{\emph{Eight coordinates, two rounds}}\\
All deterministic policies & $0.09184633$ & --- & none\\
One-fair-bit two-schedule mixture & $0.07671887$ & --- & 1 fair bit\\
Uniform three-schedule mixture & $0.07650248$ & --- & $\log_2 3$ bits\\
\bottomrule
\end{tabular}
\caption{Divergence in nats. Exact-oracle values in the deterministic rows, the geometric four-schedule row and the eight-coordinate rows are rounded centers of certified rational intervals (Appendix~\ref{app:cert}); the deterministic rows cover every deterministic rule within the round budget and every product kernel. The other oracle values and all learned values are floating-point evaluations of complete output laws; learned values are mean $\pm$ sample standard deviation over three training seeds.}
\label{tab:learned}
\end{table}

Every trained predictor stays strictly below the certified deterministic optimum, so the finite separation survives learning the conditionals. At this size the gain does not hinge on the geometry: a random four-schedule library, drawn once and not selected by its KL, already beats every deterministic policy with the same two random bits, and more scheduling randomness helps further; the geometric library carries the asymptotic guarantee of Theorem~\ref{thm:three}. Appendix~\ref{app:controls} reports controls for unequal weights and for a target without a shared component.

\section{Discussion}
The separation arises from a constraint on how information is used during sampling. A deterministic first batch must serve every possible dependence pattern, and a fixed number of remaining rounds forces some dependence to be discarded. A randomized choice instead distributes its output across the patterns. The label is hidden in both cases, and all conditional queries refer to the same target; randomization therefore buys accuracy that no amount of adaptivity can recover.

Several questions remain open. What is the best attainable mixture error at a fixed depth, and is the exponent of the present construction optimal? How small can a successful library be in general? Which distributions admit useful component-specific schedules, and can such libraries be learned from samples of an unknown target? Finally, how does the separation evolve when the number of rounds grows with the dimension, or when the sampling interface allows joint batch kernels, revision or rejection?

\section*{AI Statement}
Part of the theoretical derivations was carried out with Codex (OpenAI), and Claude (Anthropic) was used to polish the writing. The authors take full responsibility for the content.

\begingroup\small\raggedright

\endgroup

\clearpage
\appendix
\section{Additional three-round bounds}\label{app:three}
\subsection{Total variation}
Call a line bad if its first-round occupancy is not one. Every subset of $V$ has at least three bad lines across the four directions. If its size is not three, each parallel class has a bad line. For three collinear points their own parallel class has three bad lines; for three noncollinear points the pair directions give six. Across blocks, some direction therefore has at least $3r/4$ bad lines.

For a pair law with probability matrix $A$, the maximum affinity with a product law is $\|\sqrt A\|_{\mathrm{op}}$. Unconditional pairs have maximum $\beta$, conditional pairs have maximum $\sigma$, with $\beta\le\sigma<1$, and a triple's affinity is no larger than that of a marginal pair. Fix a component $z$. At every node of the deterministic reveal tree, the transition of $P_z^{\otimes r}$ factors across the lines of direction $z$. Their affinity with the product sampler transition is at most $\sigma^{c(h)}$, where $c(h)$ counts lines contributing at least two fresh coordinates. Each bad line forces such a node on every complete path.

Multiplying square-root transition weights along the tree gives
$\Aff(P_z^{\otimes r},q_{A,K})\le\sigma^{3r/4}$. To justify the multiplication without independence of adaptive decisions, normalize the square-root weights at each node; the affinity of the two output laws is then the expectation of the product of these normalizing masses, and every path product has the asserted bound. Finally, for $p=\nu P+(1-\nu)R$,
\begin{equation}\label{eq:tv-mixture}
\TV(p,q)\ge\nu-\sum_x\min\{\nu P(x),q(x)\}
\ge\nu-\sqrt\nu\Aff(P,q).
\end{equation}
Taking $\nu=1/4$ proves \eqref{eq:three-tv}.

\subsection{The deterministic optimum}
Choose two first-round points in each block, with their directions cycling through the four possibilities. Let $A$ be the resulting $2r$ observations and $\widehat Z$ their maximum-posterior direction. In round two reveal one remaining point from each estimated line, and reveal all remaining data in round three. The batch sizes are $2r,3r,4r$.

The first-observation component laws $F_z$ satisfy
\[
\Aff(F_z,F_w)=\beta^{n_z+n_w}\le\beta^{2\lfloor r/4\rfloor},
\]
where $n_z$ counts first pairs in direction $z$. Consequently, with $e_r=3\beta^{2\lfloor r/4\rfloor}$,
\[
\Pr_p(\widehat Z\ne Z)\le e_r,\qquad H(Z\mid A)\le e_r.
\]
The first inequality is a likelihood-ratio square-root bound; the second follows from the logarithm inequality used in Lemma~\ref{thm:posterior}.

For a true direction $z$, a block whose first pair is parallel to $z$ contributes $a+2b$ when the subsequent batches complete the lines of $z$; every other block contributes $b$. The conditional pair mutual information is $b$ for either value of the conditioning bit. Thus the componentwise cost on every correctly classified path is the constant
\[
L_z=rb+(a+b)n_z.
\]
Let $\mathcal J_z(X)$ denote the cumulative conditional total-correlation cost of the constructed adaptive rule under $P_z^{\otimes r}$. On a misclassified path, its first-round cost is $a n_z$ and its remaining cost is at most $7r\log2$. Since $a n_z\le L_z$, the pointwise bound
\[
\mathcal J_z(X)\le L_z+7r\log2\,\ind\{\widehat Z\ne z\}
\]
holds. Writing $J=\frac14\sum_z\E_{P_z^{\otimes r}}\mathcal J_z(X)$ and using $\sum_z n_z=r$ gives
\[
J\le\frac14\sum_z L_z+7r\log2\,\Pr(\widehat Z\ne Z)
\le r\gamma+7r\log2\,e_r.
\]

Use exact marginals of $p_r$ for this rule. Its first-round scalar kernels are fair under every component. Let $H_3$ be the committed history before the third round and $M$ the average componentwise scalar-kernel mismatch in \eqref{eq:tree}. The pointwise posterior bound from Lemma~\ref{thm:posterior} and the fixed batch sizes $3r,4r$ give
\[
M\le3rH(Z\mid A)+4rH(Z\mid H_3)
\le7rH(Z\mid A)\le7re_r.
\]
The history $H_3$ refines $A$, since every reveal set is a deterministic function of previously committed values. Finally, \eqref{eq:discount} and $H(Z\mid X)\le H(Z\mid A)$ yield
\[
\KL(p_r\|q_A)=J+M-\log4+H(Z\mid X)
\le r\gamma-\log4+[7r(1+\log2)+1]e_r.
\]
This proves \eqref{eq:three-sharp}.

\begin{figure}[t]
\centering\includegraphics{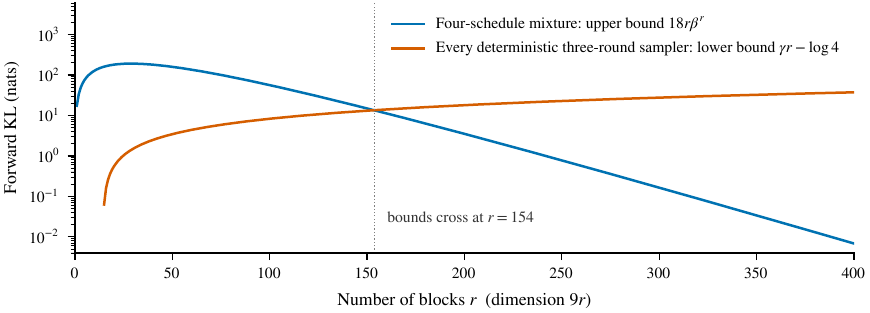}
\caption{The analytic bounds of Theorem~\ref{thm:three}: the upper bound $18r\beta^r$ for the four-schedule mixture and the lower bound $\gamma r-\log4$ for every deterministic three-round sampler. The upper bound lies below the lower bound from $r=154$ on; the exact certificate \eqref{eq:nine} already separates the two classes at $r=1$.}
\label{fig:bounds}
\end{figure}

\section{Proof of the two-round separation}\label{app:two}
\begin{proof}
An edge is bad when both endpoints lie in the same round. For a first set $C$ in a four-bit block, the average number of bad edges over the three matchings is $2,1,2/3,1,2$ at sizes $0,1,2,3,4$. Each bad edge contributes $I_\eta$ under its component. Summing over $2r$ blocks and applying Lemma~\ref{lem:latent} proves the KL bound. Some matching has at least $4r/3$ bad edges, and the reveal-tree affinity argument with \eqref{eq:tv-mixture} gives the TV bound.

For a designated matching $s$, write its two wrong matchings as $w_1,w_2$. In the first $r$ blocks reveal the two endpoints of any fixed edge of $w_1$; in the other $r$ blocks do the same for $w_2$. Each selected pair crosses the designated matching, so its law is uniform and the coordinates of the second round are conditionally independent. Each wrong matching produces a correlated first pair in $r$ blocks and uniform first pairs in the others. Lemma~\ref{thm:posterior} with $N_s=4r$ and $b_s=2\beta_\eta^r$ proves the mixture bound.

The rational certificate for \eqref{eq:eight} is in Appendix~\ref{app:cert}. A one-round exact-marginal sampler has the output law $\prod_jp(X_j)$, independent of scheduling randomness.
\end{proof}
The one-bit finite witness uses two schedules; the exponential bound uses all three.

\section{Proof for every fixed depth}\label{app:depth}
\begin{proof}
A part of the true partition $z$ is bad if its first-round occupancy is not one. Occupancy zero forces a simultaneous pair in the remaining $T-1$ rounds; occupancy at least two incurs such a pair immediately. Each bad part costs at least $a_T$ pathwise. For a first set $C$, the average number of bad parts is at least one unless $|C|=0$ or $2$. At size zero it is two. At size two, both parts are bad exactly when the two coordinates lie in the same part, an event of probability $(T-1)/(2T-1)$; hence the average is at least $2(T-1)/(2T-1)$. Summation and Lemma~\ref{lem:latent} prove \eqref{eq:allTlower}.

Fix a designated partition $s$. For any distinct partition $w$, some pair crosses $s$ but lies within one part of $w$; otherwise the two partitions would coincide up to exchanging their parts. Cycle through the $K-1$ wrong partitions over blocks and reveal such a pair first. In each remaining round reveal one remaining coordinate from each $s$-part. This factorizes the designated component, and every scalar marginal is fair. Under each wrong component, the law of the first batch has affinity at most $\beta^{\lfloor r/(K-1)\rfloor}$ with its law under the designated component, since the first batch contains that many pairs internal to the wrong component. Apply Lemma~\ref{thm:posterior} with $2(T-1)r$ remaining coordinates.
\end{proof}

\section{Proof of the logit perturbation bound}\label{app:logit}
\begin{proof}
The Bernoulli log partition $A(b)=\log(1+e^b)$ satisfies $A''\le1/4$. Its Bregman remainder is at most $e^2/8$, giving
\[
\KL(u\|\widehat v)\le\KL(u\|v)+(v-u)e+e^2/8.
\]
Pinsker's inequality gives $|u-v|\le\sqrt{\KL(u\|v)/2}$, and
$\sqrt{D/2}|e|\le D+e^2/8$. Substitution proves the claim.
\end{proof}

\section{Exact finite certificates}\label{app:cert}
\subsection{The eight-bit target and schedules}
Index the vertices in each four-bit block by $0,1,2,3$. At $\eta=1/4$, an agreeing edge has mass $3/8$ and a disagreeing edge mass $1/8$. For two blocks, the mixture has rational masses
\[
p(x)=\frac1{12288}\sum_{s=0}^2 3^{a_s(x)},
\]
where $a_s(x)$ is the number of agreeing edges among the four edges of matching $s$ in the two blocks. Three first batches are
\[
C_0=\{0,2,4,7\},\quad C_1=\{0,1,4,7\},\quad C_2=\{0,1,4,6\},
\]
and their complements form round two. The one-bit witness uses $C_0$ or $C_2$, each with probability $1/2$. All kernels are exact conditionals of the marginal mixture.

Write $D=12288$ and $w_x=Dp(x)$. For each first set $C$, let $\mu(C,a)$ be the corresponding integer marginal mass. The scalar cross-entropy numerator is
\[
|C|D\log2+\sum_a\sum_{j\notin C}
\left\{\mu(C,a)\log\mu(C,a)-\sum_{u=0}^1\mu(C\cup\{j\},a,u)\log\mu(C\cup\{j\},a,u)\right\}.
\]
Subtract $D\log D-\sum_xw_x\log w_x$ and divide by $D$. Taking the minimum over all nonempty $C$ gives the two-round optimum; by \eqref{eq:tree}, exact marginal kernels also minimize over arbitrary product kernels for each such rule.

The program \texttt{certify\_eight\_bit.py} encloses the optimum in
\[
[0.091846330389724180153102,\;0.091846330389724180153103],
\]
the one-bit mixture in
\[
[0.076718867510585191046361,\;0.076718867510585191046362],
\]
and the three-schedule mixture in
\[
[0.076502475313345167770171,\;0.076502475313345167770172].
\]
The decimal endpoints are rounded outwards.

\subsection{The nine-bit adaptive computation}
The target is \eqref{eq:dataonly} at $r=1$, with denominator $16384$. The computation has $3^9=19683$ partial states, 242461 second-stage actions and 1933442 action--outcome branches. A ternary state records whether each coordinate is unrevealed, zero or one. Let $\mu(h)$ be its integer mass and $V_t(h)$ the minimum remaining scalar cross-entropy multiplied by that mass. Completed states have value zero, and one remaining round forces all unrevealed coordinates to be drawn. For a nonempty reveal set $B$, the interval recursion evaluates
\begin{align*}
\mathcal C_t(h,B)
={}&\sum_{j\in B}\left[\mu(h)\log\mu(h)-\sum_{u=0}^1\mu(h,j=u)\log\mu(h,j=u)\right]\\
&+\sum_{y_B}V_{t-1}(h,y_B),\qquad
V_t(h)=\min_B\mathcal C_t(h,B).
\end{align*}
The minimum of all lower endpoints is a lower bound on the true minimum and the minimum of all upper endpoints an upper bound; the two minima need not be attained by the same action. A partial assignment and remaining budget determine both the conditional target and the feasible continuations, so retaining the earlier reveal order cannot improve the optimum.

\texttt{certify\_nine\_bit.py} reconstructs the target masses, generates rational log enclosures and evaluates the Bellman recursion in arbitrary-precision integer arithmetic. After subtracting the target entropy, the optimum is enclosed by
\[
[0.078112688714513902569512,\;0.078112688714513902569513].
\]
Independent rational reconstruction of all four sampler laws gives mixture KL in
\[
[0.070229296322072692413195,\;0.070229296322072692413196].
\]
The random choice of schedule fixes all three batches in advance, and every kernel in a batch uses the same history from before the round.

\subsection{Certified logarithms}
For $1\le y\le2$, let $t=(y-1)/(y+1)\in[0,1/3]$. Then
\[
\log y=2\sum_{j=0}^{N-1}\frac{t^{2j+1}}{2j+1}+R_N,
\qquad
0\le R_N\le\frac{2t^{2N+1}}{(2N+1)(1-t^2)}.
\]
All series terms and the remainder bound are evaluated as exact rational numbers with $N=48$. Power-of-two range reduction reduces any positive integer logarithm to this interval, and subtracting two integer log intervals encloses the logarithm of any positive rational. Both certificates use a $10^{-36}$ grid with arbitrary-precision integer arithmetic, and the mixture calculations retain exact rational probabilities throughout.

\subsection{Random-library control}\label{app:randomlibrary}
In the coordinate numbering used by our programs, the four control schedules are
\[
\begin{array}{ccc}
\{2,5,8\} & \{0,3,4\} & \{1,6,7\}\\
\{1,3,4\} & \{2,6,7\} & \{0,5,8\}\\
\{1,2,3\} & \{6,7,8\} & \{0,4,5\}\\
\{1,6,7\} & \{2,5,8\} & \{0,3,4\}.
\end{array}
\]
Rows are schedules and columns are successive batches; points of $\F_3^2$ are ordered lexicographically, starting with index zero. The library is drawn offline with seed 20261016 as four distinct uniformly drawn partitions, and the first generated library is reported. It uses the same predictor checkpoints and the same two-bit selection rule as the geometric library.

\section{A local criterion for deterministic error}\label{app:local}
The explicit geometry can be separated from the lower-bound argument. Suppose $P_z$ is a full-support distribution on one block of $m$ coordinates. Let $c_z(C,T)$ be the minimum expected total-correlation cost under $P_z$ among local reveal rules with first local batch $C$ and $T$ global time slots; empty local slots are allowed, and the component index is granted to this local optimization. Define
\[
\kappa_T=\min_{C\subseteq[m]}\sum_z\nu_zc_z(C,T).
\]
For $p_r=\sum_z\nu_zP_z^{\otimes r}$, every deterministic product sampler with round budget $T$ satisfies
\begin{equation}\label{eq:localgeneral}
\KL(p_r\|q)\ge r\kappa_T-H(\nu).
\end{equation}
To see this, apply \eqref{eq:tree} under each component. Conditional on the latent index and a consistent history, blocks remain independent, so their conditional total correlations add. Fixing the complete values of all other blocks induces a deterministic local rule with the same first set and at most $T$ slots; those values are independent of the current block under $P_z^{\otimes r}$. Its expected cost is at least $c_z(C,T)$, and averaging the external values preserves the bound. Summing over blocks and invoking Lemma~\ref{lem:latent} proves \eqref{eq:localgeneral}.

The quantity $\kappa_T$ is finite and checkable. For fixed $m,T,K$ the local policy set is finite and its objective consists of continuous entropy expressions, so $\kappa_T$ is continuous on full-support component laws, and $\kappa_T>0$ persists under small perturbations of the components.

\section{Proof of the schedule-weight bounds}\label{app:weights}
For $s\ne z$, marginalizing to the first batch of schedule $s$ gives
\begin{equation}\label{eq:offdiag}
\Aff(P_z^{\otimes r},q_{s,K_s,r})\le\beta^r.
\end{equation}
Indeed, under $P_z$ that batch contains a correlated pair with matrix
$A=\left(\begin{smallmatrix}3/8&1/8\\1/8&3/8\end{smallmatrix}\right)$ and one independent fair bit, while its sampler law is a coordinate product. For any product law on the pair, the affinity is $x^\top\sqrt A\,y$, where $x,y$ are nonnegative unit vectors, with maximum
$\|\sqrt A\|_{\mathrm{op}}=\sqrt{3/8}+\sqrt{1/8}=\beta$.
The singleton contributes at most one; multiplication over blocks and affinity data processing prove \eqref{eq:offdiag}, independently of every later kernel.

Pointwise subadditivity of the square root gives, for any probability vector $\pi$, including those with zero entries,
\[
\Aff(P_z^{\otimes r},q_{\pi,K,r})
\le\sqrt{\pi_z}+\beta^r\sum_{s\ne z}\sqrt{\pi_s}.
\]
Jensen's inequality yields $\KL(P\|Q)\ge-2\log\Aff(P,Q)$, and Lemma~\ref{lem:latent} therefore proves
\begin{equation}\label{eq:weights-unified}
\KL(p_r\|q_{\pi,K,r})\ge
-\frac12\sum_z\log\left(\sqrt{\pi_z}+\beta^r\sum_{s\ne z}\sqrt{\pi_s}\right)-\log4.
\end{equation}
For positive $\pi$, factoring $\sqrt{\pi_z}$ out of each logarithm gives \eqref{eq:weights-lower}. Joint-law data processing and the per-component posterior bound give
\[
\KL(p_r\|q_{\pi,\mathrm{exact},r})
\le\KL(u_4\|\pi)+\frac14\sum_z\KL(P_z^{\otimes r}\|q_{z,\mathrm{exact},r})
\le\KL(u_4\|\pi)+18r\beta^r,
\]
which is \eqref{eq:weights-upper}; $\varepsilon_r(\pi)\to0$ gives the stated limit for each fixed positive $\pi$.

If the support $S$ of $\pi$ has size $m\le3$, each $z\notin S$ satisfies
\[
\Aff(P_z^{\otimes r},q_{\pi,K,r})
\le\beta^r\sum_{s\in S}\sqrt{\pi_s}\le\sqrt m\,\beta^r.
\]
Using the divergence bound for these $4-m$ components and nonnegativity for the others, the component identity gives \eqref{eq:missing-weight}, uniformly over all supported weights and scalar kernels.

\section{Proof of stability under component perturbations}\label{app:perturb}
We first bound the change in the cost of an arbitrary local adaptive rule. For distributions $P,Q$ on a pair $(U,V)$ with $t=\TV(P,Q)$,
\begin{equation}\label{eq:entropy-cont}
|H_P(U\mid V)-H_Q(U\mid V)|\le t\log|\mathcal U|+h_2(t).
\end{equation}
For $0<t<1$, write $P=(1-t)R+tP_1$ and $Q=(1-t)R+tQ_1$ using their common part and introduce the Bernoulli mixture indicator $J$. Conditional entropy equals the weighted component conditional entropies plus $I(U;J\mid V)\in[0,h_2(t)]$, and the residual conditional entropies lie in $[0,\log|\mathcal U|]$, proving \eqref{eq:entropy-cont}; the endpoints are immediate. Sharper continuity bounds are known \citep{alhejji}.

Fix a deterministic reveal tree $A$ on $m$ bits, and let $H_j^A(X)$ record the committed history immediately before coordinate $j$ is drawn. A completed vector determines that history, so the map is the same under any reference law, and each feasible history is the cylinder specified by its committed values. The expected sum of conditional total correlations is therefore
\[
J_P(A)=\sum_{j=1}^m H_P(X_j\mid H_j^A)-H_P(X).
\]
If $\TV(P,Q)\le\epsilon\le1/2$, applying \eqref{eq:entropy-cont} to each induced pair $(X_j,H_j^A)$ and to $X$ itself, and using that total variation decreases under deterministic maps, gives the uniform bound
\begin{equation}\label{eq:cost-cont}
|J_P(A)-J_Q(A)|\le
2m\epsilon\log2+(m+1)h_2(\epsilon)=:\omega_m(\epsilon).
\end{equation}
The bound also holds for local trees with empty time slots and survives taking a minimum over any fixed set of trees.

For the components $P_z$, the local quantity $\kappa_3$ of Appendix~\ref{app:local} is at least $\gamma$, so \eqref{eq:cost-cont} gives
$\widetilde\kappa_3\ge\gamma-\omega_9(\epsilon)$.
Conditional on $z$, the perturbed blocks remain independent, and fixing all values outside a block induces a local tree with the same first set, so the argument for \eqref{eq:localgeneral} applies. Summing the costs over blocks and averaging over components gives at least $r\widetilde\kappa_3$, and the entropy discount of at most $\log4$ proves \eqref{eq:perturb-lower}, including arbitrary scalar kernels. Full support follows since the smallest atom of $P_z$ is $1/4096$ and each atom changes by at most $\epsilon<1/4096$.

For the upper bounds, keep the original schedule assigned to each component. Its unperturbed local cost is zero, so \eqref{eq:cost-cont} bounds its perturbed cost by $\omega_9(\epsilon)=\omega(\epsilon)$, and the $r$ blocks multiply the cost by $r$. Every perturbed scalar marginal, and its mixture marginal, lies in $[1/2-\epsilon,1/2+\epsilon]$, so
\[
\KL(\Ber(u)\|\Ber(v))\le\frac{(u-v)^2}{v(1-v)}
\le\frac{4\epsilon^2}{1/4-\epsilon^2}.
\]
With three coordinates in each first slice, the first-round mismatch is at most $r\alpha(\epsilon)$. For probability vectors $U,\widetilde U$,
$\|\sqrt U-\sqrt{\widetilde U}\|_2\le\sqrt{2\TV(U,\widetilde U)}$, so changing each of two prefix laws by at most $\epsilon$ in TV changes their affinity by at most $2\sqrt{2\epsilon}$, by Cauchy--Schwarz. Every wrong-component prefix affinity is consequently at most $\rho(\epsilon)$ per block and $\rho(\epsilon)^r$ across blocks, and Lemma~\ref{thm:posterior} gives a later mismatch at most $6r\cdot3\rho(\epsilon)^r$. Adding the dependence cost and the first-round mismatch proves \eqref{eq:perturb-upper}. Finally, applying the logit perturbation inequality only to the scalar-kernel terms in \eqref{eq:tree}, the dependence cost is paid once while the scalar mismatch bound is doubled; averaging over components gives \eqref{eq:perturb-learned}.

\paragraph{A numerical case.}
At $\epsilon=10^{-5}$, rational logarithm enclosures and rational square comparisons certify
\[
\begin{gathered}
\gamma>0.0966,\quad\omega<0.001377,\quad\alpha<5\cdot10^{-9},\\
\rho<39/40,\quad(39/40)^{256}<0.00154,\quad\log4<1.387.
\end{gathered}
\]
Consequently, for every $r\ge256$, the deterministic lower bound divided by $r$ exceeds
\[
0.0966-0.001377-1.387/256>0.0898.
\]
If $\widetilde{\mathcal E}_2\le9r/100$, the randomized upper bound divided by $r$ is below
\[
0.001377+10^{-8}+36(0.00154)+9/400<0.0794.
\]
The bounds also apply to smaller $\epsilon$ and error budgets, since the relevant upper bounds increase on this interval. The program \texttt{certify\_perturbation.py} verifies each strict inequality in rational arithmetic.

\section{Conditional-predictor experiment}\label{app:training}
We distilled the exact marginal logits of the nine-bit target into a shared MLP. Its 18 inputs are the observed values, encoded as $-1,0,1$, and nine observation indicators; two width-128 GELU layers produce nine logits, for 20105 parameters in total. Each of three training seeds controls the initialization and the sample and mask streams. Each run uses 3000 Adam steps with 512 fresh target draws per step, independent Bernoulli$(1/2)$ observation masks and squared logit loss on unrevealed positions, with learning rate $0.003$ decayed by a cosine schedule to $0.00015$. The teacher supplies exact conditional logits computed from the full target law, and the component label is never supplied to the network.

The geometric-library KL values of the three predictors are $0.0718770186$, $0.0715184720$ and $0.0720248477$, each below the certified deterministic optimum $0.0781126887$. Their root-mean-square logit errors along the assigned schedules, $\sqrt{\mathcal E_2/9}$, are $0.02914$, $0.02969$ and $0.02997$. For comparison with a larger budget, the four-round deterministic optimum is approximately $0.0303017054$; all comparisons in the main text use three rounds.

\section{Schedule weights and shared dependence}\label{app:controls}
\begin{lemma}[Mismatched schedule weights]\label{cor:weights}
Suppose $p_r=\sum_z\nu_zP_{z,r}$, the finite component family satisfies
$\Aff(P_{z,r},P_{w,r})\to0$ for $z\ne w$, and
$\KL(P_{z,r}\|q_{z,r})\to0$. For any fixed strictly positive schedule law $\pi$,
\[
\KL\left(p_r\middle\|\sum_z\pi_zq_{z,r}\right)\longrightarrow\KL(\nu\|\pi).
\]
\end{lemma}
\begin{proof}
Joint-law data processing gives the upper bound
$\KL(\nu\|\pi)+\sum_z\nu_z\KL(P_{z,r}\|q_{z,r})$.
A maximum-posterior classifier based on the completed data identifies each component with error tending to zero by the pairwise affinity bound, and Pinsker's inequality transfers this property from $P_{z,r}$ to $q_{z,r}$. The classifier's output law therefore tends to $\nu$ under the target and to $\pi$ under the sampler, and KL data processing through the classifier gives the matching lower limit.
\end{proof}

For conditionally iid blocks the affinities of the first-batch laws multiply. If the assigned schedule factorizes its component exactly, all scalar marginals agree in the first round, and under every wrong component the law of the first batch has affinity at most $\beta<1$ per block, then \eqref{eq:posterior} is $O(r\beta^r)$.

We estimated final-law KL at $r=1,4,16,64,128,256$ and $512$ blocks, each from 10000 independent target draws with exact marginal kernels evaluated in floating point (Figure~\ref{fig:mechanism}). Because $18r\beta^r$ becomes small only for $r$ in the hundreds, the uniform-weight estimates first rise, from about $0.068$ at $r=1$ to $1.642$ at $r=16$, and then fall, to $0.235$ at $r=128$, $5.5\times10^{-5}$ at $r=256$ and $1.3\times10^{-11}$ at $r=512$, where the analytic guarantee is $18r\beta^r<1.803\times10^{-4}$. Two controls isolate the two ingredients of the construction. With schedule weights $0.7,0.1,0.1,0.1$, Proposition~\ref{prop:weights} predicts the floor $\KL(u_4\|\pi)=0.4298132$; the estimate at $r=512$ is $0.42631$ with empirical standard error $0.00845$. Redrawing the latent direction independently in every block removes the evidence pooled across blocks; with exact kernels for this altered product target and one schedule drawn for all blocks, the estimate at $r=512$ is $54.4265$ with standard error $0.0904$. A shared hidden component and uniform weights are thus both essential to the exponential accuracy.

\begin{figure}[htbp]
\centering\includegraphics{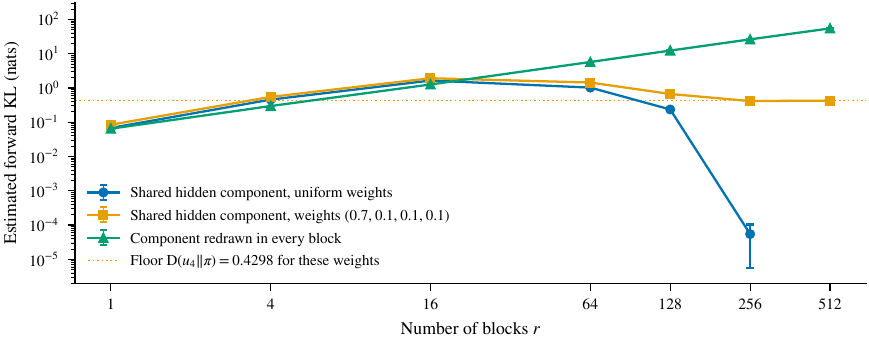}
\caption{Monte Carlo estimates of forward KL with exact marginal kernels, from 10000 target draws per point; error bars are standard errors. With a shared hidden component and uniform weights the error collapses once $r$ reaches the hundreds (the estimate at $r=512$, $1.3\times10^{-11}$, lies below the plotted range). Unequal weights level off at the floor $\KL(u_4\|\pi)$ of Proposition~\ref{prop:weights}, and redrawing the component in every block makes the error grow roughly linearly in $r$.}
\label{fig:mechanism}
\end{figure}

\end{document}